\documentclass{tlp}
 
\nonstopmode

\usepackage{amsmath}
\usepackage{amssymb}
\usepackage{tikz}
\usetikzlibrary{arrows}
\usetikzlibrary{topaths}
\usepackage{stmaryrd} 
\usepackage{ifthen} 
\usepackage{aliascnt} 
\usepackage{cleveref} 

\newcommand\wrt{w.r.t\@.}
\newcommand\st{s.t\@.}

\newcommand\Newtheorem[2]{
  \newaliascnt{#1}{theorem} 
  \newtheorem{#1}[#1]{#2}  
  \aliascntresetthe{#1} 
}

\newtheorem{theorem}{Theorem}[section]
\Newtheorem{lemma}{Lemma}
\Newtheorem{proposition}{Proposition}
\Newtheorem{corollary}{Corollary}
\Newtheorem{definition}{Definition}
\Newtheorem{example}{Example}
\Newtheorem{remark}{Remark}
\crefname{figure}{figure}{figures}

\renewcommand\cref[1]{\Cref{#1}}

\newcommand\ra{\rightarrow}

\def\lfp{l\@.f\@.p\@.}
\def\gfp{g\@.f\@.p\@.}

\newcommand\ie{i.e\@.}

\newcommand\R{\mathcal R}
\renewcommand\L{\mathcal L}

\newcommand\X{{\cal X}}
\newcommand\Y{{\cal Y}}

\newcommand\BB{{\mathbb B}}
\newcommand\CC{{\mathbb C}}
\newcommand\DD{{\mathbb D}}
\newcommand\EE{{\mathbb E}}
\newcommand\FF{{\mathbb F}}
\newcommand\GG{{\mathbb G}}
\newcommand\HH{{\mathbb H}}
\newcommand\KK{{\mathbb K}}
\newcommand\LL{{\mathbb L}}
\newcommand\PP{{\mathbb P}}

\newcommand\VV{{\mathbb V}}

\newcommand\fv{\textup{fv}}

\newcommand\CT{{\mathcal{C}}}
\newcommand\CP{\CT\P}

\newcommand\lv{\textup{lv}}

\newcommand\at{{\;@\;}}

\renewcommand\P{{\mathcal P}}
\newcommand\Q{{\mathcal Q}}

\def\simplif{~{\Longleftrightarrow}~}
\def\propag{~{\Longrightarrow}~}

\newcommand\default[2]{\ifthenelse{\equal{#1}{}}{#2}{#1}}

\newcommand\rapa[1][\P]{\,{\xrightarrow{{}_{#1}}_a}\,}
\newcommand\rapas[1][\P]{{\xrightarrow{{}_{#1}}^{_*}_{a}}}

\newcommand\estate[3][]{{{\ifthenelse{\equal{#1}{}}{}{\exists {#1}}}\left< {#2};{#3} \right>}}

\newcommand\state[3]{\langle \default{#1}{\emptyset};\default{#2}{\true}
;\default{#3}{\emptyset}
 \rangle}

\newcommand\states{\Sigma}

\newcommand\redex[4]{{#1}@\state{#2}{#3}{#4}}

\newcommand\astate[3]{\left< {#1};{#2};{#3} \right>}
\newcommand\abstractindice{{a}}
\newcommand\aequiv{\equiv_\abstractindice}

\newcommand\chr{\textup{chr}} 
\newcommand\id{\textup{id}}

\newcommand\rapc{{\xrightarrow{{}_{\P}}_c}}

\newcommand\cstate[6]{\left< \! \left< {#1}; {#2};{#3}; {#4} \right> \!\right>_{#6}^{#5}
}

\newcommand\false{\textbf{false}}
\newcommand\true{\textbf{true}}

\newcommand\0{{\state{\emptyset}{\false}{\emptyset}}}

\newcommand\smodel[1]{{\mathcal F}^{\,\CT}_\text{co}({#1})}

\newcommand\rap{\xrightarrow{_\P}}

\newcommand\lmodel[1]{{\mathcal F}^{\,\CT}({#1})}

\newcommand\hstate[3]{\left< {#1};{#2};{#3} \right> }

\newcommand\raph[1][\P]{\xrightarrow{_{#1}}_h}
\newcommand\raphs[1][\P]{\xrightarrow{_{#1}}_h^{_*}}

\newcommand\hybridindice{{h}}
\newcommand\hequiv{\equiv_\hybridindice}

\newcommand\hmodel[1]{{\mathcal F}^{\,\CT}_\text{h}({#1})}

{

\renewcommand\L{{\mathcal L}}

\newcommand\Bp[1][\Y]{{T^\CT_\P}}

\newcommand\cexists[1]{\exists_{\,\text{-}{#1}}}

\newcommand\prth[1]{{\left( {#1} \right)}}
\newcommand\prea[1][\P]{{\left[ {#1 }\right]_h}}
\newcommand\pree[1][\P]{{\left< {#1 }\right>_h}}
\newcommand\eqclass[2][h]{{\left\llbracket {#2 } \right\rrbracket_{#1}}}

\newcommand\esum[2]{\ifthenelse{\equal{#2}{}}{[#1]}{[#1|#2]}}
\newcommand\ezero{[]}
\newcommand\eone{\ezero^*}

\begin{document}

  \title[(Co-)Inductive semantics for Constraint Handling Rules]%
  {(Co-)Inductive semantics for Constraint Handling Rules%
  }
\author[R\'emy Haemmerl\'e]{R\'emy Haemmerl\'e\\
Technical University of Madrid}

\maketitle

\begin{abstract}
  In this paper, we address the problem of defining 
  a fixpoint semantics for Constraint Handling Rules (CHR) that
  captures the behavior of both simplification and propagation rules
  in a sound and complete way with respect to their declarative
  semantics. 
  Firstly, we show that the logical reading of states with respect to
  a set of simplification rules can be characterized by a least
  fixpoint over the transition system generated by the abstract
  operational semantics of CHR.  Similarly, we demonstrate that the
  logical reading of states with respect to a set of propagation rules
  can be characterized by a greatest fixpoint.
  Then, in order to take advantage of both types of rules without
  losing fixpoint characterization, we present an operational
  semantics with persistent.
  We finally establish that this semantics can be characterized by two
  nested fixpoints, and we show the resulting language is an elegant
  framework to program using coinductive reasoning.
\end{abstract}


\section{Introduction}
\label{sect:introduction}

Owing to its origins in the tradition of Constraint Logic Programming
(CLP) \cite{JL87popl}, Constraint Handling Rules (CHR)
\cite{Fruehwirth98jlp} feature declarative semantics through direct
interpretation in classical logic.  However, no attempt to provide
fixpoint semantics to the whole language, sound and complete \wrt\
this declarative semantics has succeeded so far. This is particularly
surprising considering that the fixpoint semantics is an important
foundation of the declarative semantics of CLP.  It is perhaps because
CHR is the combination of two inherently distinct kinds of rules that
this formulation is not so simple. On the one hand, the so-called
Constraint Simplification Rules (CSR) replace constraints by simpler
ones while preserving their meaning. On the other hand, the so-called
Constraint Propagation Rules (CPR) add redundant constraints in a
monotonic way.  Even though in the declarative semantics the two
notions merge, one of the main interests of the language comes from
the explicit distinction between the two. Indeed, it is well known
that propagation rules are useful in practice but have to be managed in a
different way from simplification rules to avoid trivial non-termination.
(See for instance explanations by \citeN{Fruehwirth09cambridge} or
\citeN{BRF10iclp}.)

Soundness of the operational semantics of CSR (\ie\ to each derivation
corresponds a deduction) have been proved by
\citeN{Fruehwirth09cambridge}, while completeness (\ie\ to each
deduction corresponds a derivation) has been tackled by
\citeN{AFM99constraints}.  However, it is worth noticing that the
completeness result is limited to terminating programs.  On the other
hand, the accuracy of CPR \wrt\ its classical logic semantics is only
given through its naive translation into CSR.  Since any set of propagation
rules trivially loops when seen as simplification rules, the completeness
result does not apply to CPR.

It is well known that termination captures least fixpoint (\lfp) of
states of a transition system. Quite naturally, non-termination
captures the greatest fixpoint (\gfp). Starting from this observation,
we show in this paper that if they are considered independently, CSR
and CPR can be characterized by a \lfp\ (or inductive) and \gfp\ (or
coinductive) semantics respectively, providing along the way the first
completeness result for CPR. Then, in order to take advantage of both
types of rules without losing fixpoint characterization, we present an
operational semantics $\omega_h$ similar to the one recently proposed
by \citeN{BRF10iclp}. Subsequently we demonstrate that this new
semantics can be characterized by two nested fixpoints and can be
implemented in a simple manner to provide the first logically complete
system (\wrt\ failures) for a segment of CHR with propagations rules.
We show as well this semantics yields an elegant framework for
programming with coinductive reasoning on infinite (or non-well
founded) objects \cite{BM96csli}.

\smallskip

The remainder of this paper is structured as follows:
\Cref{sect:preliminaries} states the syntax of CHR and summarizes several
semantics.  In \cref{sect:fixpoints},
we present two fixpoint semantics for CHR. We show these semantics,
which built over the transition system induced by the abstract
operational semantics of CHR, offer a 
characterization of logical reading of queries \wrt\ CSR and CPR,
respectively. %
In \cref{sect:hybrid}, we define semantics with persistent constraints
related to the one recently introduced by \citeN{BRF10iclp}. %
We prove this new operational semantics can be characterized by a
\lfp\ nested within a \gfp, and give an implementation via a
source-to-source transformation.  Finally, in
\cref{sect:applications}, we illustrate the power of the language
before concluding in \cref{sect:conclusion}.

\section{Preliminaries on CHR}
\label{sect:preliminaries}

In this section, we introduce the syntax, the declarative semantics
and two different operational semantics for CHR. In the next sections,
both operational semantics will be used, the former as theoretical
foundation for our different fixpoint semantics, and the latter as a
target to implementation purposes.

\subsection{Syntax}

The formalization of CHR assumes a language of \emph{built-in
  constraints} containing the equality $=$, $\false$, and $\true$ over
some theory $\CT$ and defines
\emph{user-defined constraints} using a different set of predicate
symbols.
In the following, we will denote variables by upper case letters, $X$,
$Y$, $Z$, \dots, 
and (user-defined or built-in) constraints by lowercase letters
$c,d,e\dots$ By a slight abuse of notation, we will confuse
conjunction and multiset of constraints, forget braces around multisets
and use comma for multiset union. We note $\fv(\phi)$ the set
of free variables of any formula $\phi$. %
The notation $\cexists{\bar X} \phi$ denotes the existential closure of
$\phi$ with the exception of variables in $\bar X$, which remain free.
In this paper, we require the non-logical axioms of $\CT$ to be {\em
  coherent formula} (\ie\ formulas of the form $\forall (\CC
\rightarrow \exists \bar Z.\DD)$, where both $\CC$ and $\DD$ stand for
possibly empty conjunctions of built-in constraints).  Constraint
theories verifying such requirements correspond to
Saraswat's {\em simple constraints systems} \citeNN{SRP91popl}.

A {\em CHR program} is a finite set of
 eponymous
rules of the form:
\[ r \at \KK \backslash \HH \simplif \GG \mid \CC, \BB \]
where $\KK$ (the {\em kept head}), $\HH$ (the {\em removed head}) are
multisets of user-defined constraints
respectively, $\GG$ (the {\em guard}) is a conjunction of built-in constraints%
, $\CC$ is a conjunction of built-in constraints, $\BB$ is
a multiset of user-defined constraints and, $r$ (the {\em rule name}) is an arbitrary
identifier assumed unique in the program.
 Rules, where both heads are empty, are
prohibited. 
 Empty kept-head can be omitted together with
the symbol $\backslash$.
 The {\em local variables} of rule are the variables occurring in the
  guard and in the body but not in the head that is $\lv(r) \!=\!
  \fv(\GG , \CC, \BB) \setminus \fv(\KK, \HH)$.
CHR rules are divided into two classes:
{\em simplification rules} if the removed head is non empty and {\em
  propagation rules} otherwise. Propagation rules can be written using
the alternative syntax: \[
r \at  \KK \propag \GG \mid \CC, \BB
\]

A {\em CHR state} is a tuple $\state{\CC}{\EE}{\bar X}$, where
$\CC$ (the {\em CHR store}) is a multiset of CHR constraints%
, $\EE$ (the {\em built-in store})
is a conjunction of built-in constraints%
, and $X$ (the {\em global variables}) is a set of variables.
In the following, $\states$ will denote the set of states and
$\states_b$ the set of {\em answers} 
(\ie\ states of the form $\state{\emptyset}{\CC}{\bar X}$).  A state
is {\em consistent} if its built-in store is satisfiable within $\CT$
(\ie\ there exists an interpretation of $\CT$ which is a model of
$\exists \CC$), {\em inconsistent} otherwise.

\subsection{Declarative semantics}
\label{subsection:preliminaries:declarative}

We state now the declarative semantics of CHR. The {\em logical
  reading} of a rule and a state is as follows:
\begin{align*}
&\text{Rule:} & &  
   \KK \backslash \HH \simplif \GG \mid \CC, \BB &&
  \forall \prth{\prth{\KK \wedge \GG} \rightarrow \prth{\HH
      \leftrightarrow \cexists{\fv(\KK, \HH)}\prth{\GG \wedge \CC
        \wedge \BB}}}\\
&\text{State:} && \state{\CC}{\EE}{\bar X} &&  \cexists{\bar X}\prth{\CC
    \wedge \EE}
\end{align*}
$\CP$, the {\em logical reading} of a program $\P$ within a constraint
theory $\CT$ is the conjunction of the logical readings of the rules
of $\P$ with the constraint theory $\CT$.

\subsection{Equivalence-base operational semantics}
\label{subsection:preliminaries:abstract}

Here, we recall the {\em equivalence-based} operational semantics
$\omega_e$ of \citeN{RBF09chr}. It is similar to the {\em very abstract}
semantics $\omega_a$ of \citeN{Fruehwirth09cambridge}, the most
general operational semantics of CHR. We prefer the former because it
includes an explicit notion of equivalence, that will simplify many
formulations.  Because this is the most abstract operational semantics we
consider in this paper, we will refer to it as the {\em abstract
  (operational) semantics}. For the sake of generality, we present it in
a parametric form, according to some sound equivalence relation.

\smallskip

We will say that an equivalence relation $\equiv_i$ is {\em
  (logically) sound} if two states equivalent \wrt\ $\equiv_i$ have
logically equivalent readings in the theory $\CT$.  The equivalence
class of some state $\sigma$ by $\equiv_i$ will be noted $\llbracket
\sigma \rrbracket_i$.  For a given program $\P$ and a sound
equivalence $\equiv_i$, 
the {\em $\equiv_i$-transition relation}, noted $\rap_i$, is the least
relation satisfying the following rules:
\[
\frac{
\prth{r \at \KK \backslash \HH \simplif \GG | \CC, \BB}  \in \P\rho  \quad 
\lv(r) \cap \bar X = \emptyset 
}{ %
\state{\KK, \HH, \DD}{ \GG \wedge  \EE}{\bar X}
 \rap_i %
\state{\CC, \KK, \DD}{ \BB  \wedge  \GG \wedge  \EE}{\bar X}
}
\quad 
\frac{\sigma_1 \equiv_i \sigma_1' \quad \sigma_1' \rap_i\sigma_2' \quad \sigma_2' \equiv_i \sigma_2}
{\sigma_1 \rap_i \sigma_2}
\]
where $\rho$ is a renaming.  If such transition is possible with
$\HH=\emptyset$, we will say that the tuple $\redex{r}{\KK}{\GG\wedge
  \EE}{\bar X}$ is a {\em propagation redex} for the state $\sigma$.
The transitive closure of the relation $(\rap_i \cup \equiv_i)$ is
denoted by $\rap_i^*$.  A {\em $\equiv_i$-derivation} is a finite or
infinite sequence of the form $\sigma_1 \rap_i \dots \rap_i \sigma_n
\rap_i \dots$ %
A $\equiv_i$-derivation is {\em consistent} if so are all the states
constituting it. %
A $\equiv_i$-transition is {\em confluent} if whenever $\sigma
\rap_i^* \sigma_1$ and $\sigma \rap_i^* \sigma_2$ hold, there exists a
state $\sigma'$ such that $\sigma_1 \rap_i^* \sigma'$ and $\sigma_2
\rap_i^* \sigma'$ hold as well.

\label{definition:preliminaries:abstract:equivalence}
The {\em abstract equivalence} is the least equivalence $\aequiv$ defined
over $\states$ verifying:
   \begin{enumerate}
  \item \label{preliminaries:abstract:equivalence:equality}
 $\state{\CC}{Y \!=\! t \wedge \DD}{\bar X} \aequiv \state{\GG[Y
      \backslash t]}{Y \!=\! t \wedge \DD}{\bar X}$
  \item \label{preliminaries:abstract:equivalence:inconsistent}
$\state{\CC}{\false}{\bar X}\aequiv \state{\DD}{\false}{\bar X}$ 
 \item  \label{preliminaries:abstract:equivalence:transformation}
 $\state{\CC}{\DD}{\bar X} \aequiv \state{\CC}{\EE}{\bar X}$ if $\CT \vDash \cexists{\fv\prth{\CC,\bar X}}\prth{
    \DD} \leftrightarrow \cexists{\fv\prth{\CC,\bar X}}\prth{\EE}$
  \item \label{preliminaries:abstract:equivalence:omission}
$\state{\CC}{\EE}{\bar X} \aequiv
  \state{\CC}{\EE}{\{Y\} \cup \bar X}$ where $Y \notin \fv(\CC, \EE)$.
  \end{enumerate}
  Note that this equivalence is logically sound 
  \cite{RBF09chr}.
For a given program $\P$, the {\em abstract transition systems} is
defined as the tuple $(\Sigma, \rapa)$.

\subsection{Concrete operational  semantics} 

This section presents the operational semantics $\omega_p$ of
\citeN{DKSD07ppdp}. In this framework rules are annotated with
explicit priorities that reduce the non-determinism in the choice of
the rule to apply. As initially proposed by \citeN{Abdennadher97cp},
this semantics includes a partial control that prevents the trivial
looping of propagation rules by restricting their firing only once on
same instances. By opposition to the abstract semantics,
we will call it {\em concrete (operational) semantics}.

\smallskip

  An {\em identified constraint} is a pair noted $c\#i$, associating
  a CHR constraint $c$ with an integer $i$. For any identified constraints,
  we define the functions $\chr(c\#i)=c$ and $\id(c\#i)=i$, and extend
  them  to sequences and sets of identified constraints.  A {\em
    token} is a tuple $(r, \bar \i)$, where $r$ is a rule name and
  $\bar \i$ is a sequence of integers. 
  A {\em concrete CHR state} is a tuple of the form
  $\cstate{\CC}{\DD}{\EE}{T}{\bar X}{n}$ where $\CC$ is a multiset of CHR
  and built-in constraints, $\DD$ is a multiset of identified
  constraints, $\EE$ is a conjunction of built-in constraints, $T$ is
  a set of tokens and $n$ is an integer.  We assume moreover that the
  identifier of each identified constraints in the CHR store is unique
  and smaller than $n$.  For any program $\P$, the {\em concrete
    transition} relation, $\rapc$, is defined as following:
  \begin{description}
  \item[\textbf{Solve}] 
$\cstate{c, \CC}{\DD}{\EE}{T}{\bar X}{n} \rapc
  \cstate{\CC}{\DD}{c \wedge \EE}{T}{\bar X}{n} $\\
if  $c$ is a built-in constraint and 
$\CT \! \vDash \forall((c \wedge \BB) \leftrightarrow \BB')$.
\item[\textbf{Introduce}] 
$\cstate{c, \CC}{\DD}{\EE}{T}{\bar X}{n} \rapc
  \cstate{\CC}{c\#n, \DD}{\EE}{T}{\bar X}{n\!+\!1} $\\
if  $c$ is  a CHR constraint.
\item[\textbf{Apply}] $\cstate{\emptyset}{\KK, \HH, \EE}{\CC}{T}{\bar X}{n} \rapc
  \cstate{\BB, \GG}{\KK, \DD}{\CC \wedge \theta}{t \cup T }{\bar X}{n} $\\
  if $(p :: r \at \KK' \backslash \HH' \Leftrightarrow \GG \mid \BB)$
  is a rule in $\P$ of priority $p$ renamed with fresh variables,
  $\theta$ is a substitution such that $\chr(\KK) = \KK'\theta$,
  $\chr(\HH) = \HH'\theta$, $t = (r, \id(\KK, \HH))$, $t \notin T$,
  $\CC$ is satisfiable within $\CT$, and $\CT \vDash \forall(\CC
  \rightarrow \exists( \theta \wedge G))$. Furthermore, no rule of
  priority bigger than $p$ exists for which the above conditions hold.
\end{description}

\section{Transition system semantics for pure CSR and CPR}
\label{sect:fixpoints}

In this section, we propose a fixpoint semantics for both CSR and CPR
programs. We call it transition system semantics because it is defined
as a fixpoint over the abstract transition system, built in a way
similar to $\mu$-calculus formula \cite{CGP00mit}.
  The proofs of this section are only sketched. Detailed versions can
  be found in a technical report~\cite{Haemmerle11clip}.

\smallskip

Before formally introducing the semantics, we recall some standard
notation and results about fixpoints in an arbitrary complete lattice
$(\L, \supset, \cap, \cup, \top, \bot)$.
\footnote{For more details
  about fixpoints, one can refer, for example, to
  \citeANP{Lloyd87springer}'s Book \citeNN{Lloyd87springer}.}.
  A function $f : \L \ra \L$ is {\em monotonic} if $f(\X)
  \supset f(\Y)$ whenever $\X \supset \Y$.  An element $\X \in \L$ is
  a {\em fixpoint} for $f:\L \ra \L$ if $f(\X) =\X$. The {\em least
    fixpoint} (resp\@. {\em the greastest fixpoint}) of $f$ is its
  fixpoint $\X$ satisfying $\Y \supset \X$ (resp\@.  $\Y \subset \X$)
  whenever $\Y$ is a fixpoint for $f$. It is denoted by $\mu \X.f(\X)$
  (resp\@. $\nu \X.f(\X)$).
 \citeANP{Tarski55}'s  \citeNN{Tarski55} celebrated
  fixpoint theorem  ensures that monotonic functions have both a least
  and a greatest fixpoint.

\subsection{Inductive semantics for
  CSR}\label{sect:fixpoint:transition}

In this section, we give a first fixpoint semantics limited to CSR. It
is call inductive, since it is defined as a lfp.

\begin{definition}[Inductive transition system semantics for CSR]
\label{definition:existential:cause:operator}
For a given program $\P$ and an given sound equivalence relation
$\equiv_i$, the {\em existential immediate cause operator}  $\left<
  \P \right>_i:2^\Sigma \ra 2^\Sigma$ is defined as:
 \[
 \left< \P \right>_i(\X) = \{ \sigma \in \Sigma \mid \text{there
   exists } \sigma' \in \Sigma \text { such that } \sigma \rap_i
 \sigma' \text{ and } \sigma' \in \X\}
 \]
 The {\em inductive (transition system) semantics} of a CSR program
 $\Q$ is the set:
\[
\lmodel{\Q} =
    \mu \mathcal X. \prth{\left<\Q \right>_a\prth{\mathcal X}
      \cup \prth{\Sigma_{b} \setminus \llbracket\0\rrbracket_a }}
\]
\end{definition}

The existential immediate cause operator being clearly monotonic,
Tarski's theorem ensures the inductive semantics of a CSR program is
well defined. For a given program $\P$, $\lmodel{\P}$ is exactly the
set of states that can be rewritten by $\P$ to a consistent
answer. Remark that since answers cannot be rewritten by any program,
any state in $\lmodel{\P}$ has at least one terminating
derivation. 

\begin{example}
  \label{example:fixpoint:inductive}
Consider the program $\P_1$ consisting of the two following rules:
\begin{align*}
  a, a \simplif & \true & b \simplif& b & c \simplif& \false 
\end{align*}
We can pick the sets of consistent states of the respective form
\(\state{a^{2i}}{\CC}{\bar X}\) and \(\state{a^{2 i}, b^j}{\CC}{\bar
  X}\), where $d^n$ denotes $n$ copies of a constraint $d$.  Both
sets are fixpoints of $\lambda \mathcal X. \prth{\left<\P_1
  \right>_a\prth{\mathcal X} \cup \prth{\Sigma_{b} \setminus
    \llbracket\0\rrbracket_a }}$, but only the former is the least
fixpoint. Note, that no state of the form \(\state{a^{2i\text{+}1},
  \EE}{\CC}{\bar X}\) or \(\state{c, \EE}{\CC}{\bar X}\) are in such
fixpoints.
\end{example}

We next present a theorem that uses fixpoints semantics to reformulate
results on CSR logical semantics
\cite{Fruehwirth98jlp,AFM99constraints}.  It says that a state that
has at least one answer is in the inductive semantics of a confluent
CSR program if and only if its logical reading is satisfiable within
the theory $\CP$.
Notice that because in the context of this paper, we do not require
$\CT$ to be ground complete, we have to content ourselves with
satisfiability instead of validity.  Nonetheless, we will see in
\cref{sect:applications}, that satisfiability is particularly useful
to express coinductive definitions such as bisimulation.

\begin{theorem}
\label{theorem:fixpoint:csr}
Let $\P$ be program such that $\rapa$ is confluent.  Let
$\state{\DD}{\EE}{\bar X}$ be a state having at least one answer. We have
\[
 \state{\DD}{\EE}{\bar X} \in \lmodel{\P} %
 \text{ if and only if } %
\exists \prth{\DD \wedge \EE} \text{ is satisfiable within } \CP.
\]
\end{theorem}

\begin{proof}[Proof sketch]
  As we have said previously, $\lmodel{\P}$ is the set of states
  that can be rewritten to a consistent answer. Hence it is sufficient
  to prove:
\[
\astate{\EE}{\CC}{\bar X}  \rapas[\P] \astate{\emptyset}{\DD}{\bar Y} \text{ with  } \CP\nvDash \neg \exists \prth{\DD}
  \text{ if and only if }
\CP\nvDash \neg \exists \prth{\EE \wedge \CC}
 \]
or equivalently, the contrapositive:
\[
\CP \vDash \neg \exists \prth{\EE \wedge \CC} \text{ if and only if } \astate{\EE}{\CC}{\bar X}  \rapas[\P] \0  
\]
``If'' and ``only if'' directions are respective corollaries of
soundness and completeness for CHR (Lemma 3.20 and Theorem 3.25 in
\citeANP{Fruehwirth09cambridge}'s book
\citeNN{Fruehwirth09cambridge}).
\end{proof}

Our inductive semantics for CSR has strong connections with the
fixpoint semantics of \citeN{GM09tocl}. In contrast to ours, this
semantics focuses on input/output behaviour and is not formally
related to logical semantics, although it is constructed in similar
way as a \lfp\ over the abstract transition system. However, because
it does not distinguish propagation from simplification rules, this
semantics cannot characterize reasonable programs using
propagations
. Indeed, it has been later extended to
handle propagation rules by adding into the states an explicit token store
{\em \`a la} \citeN{Abdennadher97cp} in order to remember the
propagation history \cite{GMT08chr}. Nonetheless, such an extension
leads to a quite complicated model which is moreover incomplete
\wrt\ logical semantics.

\subsection{Coinductive semantics for CPR}
\label{sect:fixpoint:coinductive}

We continue by giving a similar characterization for CPR.  This
semantics is defined by the {\gfp} of a universal version of the cause
operator presented in \cref{definition:existential:cause:operator}.
Hence we call it coinductive.

\begin{definition}[Coinductive transition system semantics for CPR]
  For a given program $\P$ and an given sound equivalence relation
  $\equiv_i$, the {\em universal (immediate) cause operator} $\left[ \P
  \right]_i:2^\Sigma \ra 2^\Sigma$ is defined as:
  \begin{align*}
    \left[ \P \right]_i(\X) = \,&\{ \sigma \in \Sigma \mid \text{for all } \sigma' \in \Sigma,  \sigma \rap_i \sigma' \text{ implies }   \sigma' \in \X\}
  \end{align*}
  The {\em coinductive (transition system) semantics}
  of a CPR program  $\Q$ is the set:
 \[\smodel{\Q} =
 \nu \mathcal X. \prth{\left[ \Q\right]_a \prth{\X} \cap \prth{\Sigma \setminus \llbracket \0 \rrbracket_a \!}}
\]
\end{definition}

Note that contrary to the inductive semantics, the coinductive one is
not just a reformulation of the existing semantics. Indeed the
universal essence of the operator $\left[ \Q \right]_i$ conveys that
the meaning we give to CPR states relies on all of its derivations,
whereas the existential essence of the operator $\left< \Q \right>_i$
makes explicit the fact that the classical meaning states dependent on
existence of a successful derivation. This semantic subtlety is
fundamental for CPR completeness
(\cref{proposition:cpr:completeness}).

As it is the case for $\left< \Q \right>_i$, the operator $\left[ \Q
\right]_i$ is obviously monotonic. Our semantics is therefore well
defined.  Notice, that $\smodel{\Q}$ is precisely the set of states
that cannot be rewritten to an inconsistent states.  As illustrated by
the following example, states belonging to $\smodel{\Q}$ have in
general only non-terminating derivations \wrt\ the abstract
operational semantics.

\begin{example}
  \label{example:fixpoint:coinductive}
  Let $\CT$ be the usual constraints
  theory
  over integers and $\P_2$ be the program:
\begin{align*}
  q(X) \propag & q(X+1) & q(0) \propag& \false 
\end{align*}
The greatest fixpoint of $\lambda \X. \prth{\left[ \P_2\right]_a
  \prth{\X} \cap \prth{\Sigma \setminus \llbracket \0 \rrbracket_a
    \!}}$ is the set of consistent states that does not contain a CHR
constraint $p(X)$ where $X$ is negative or null. Note the empty set is
also a fixpoint but not the greatest.  Note that states such as
$\state{q(1)}{\true}{}$, which are in the greatest fixpoint, have only
infinite derivations.
\end{example}

We give next a theorem that states the accuracy of the
coinductive semantics \wrt\ the logical reading of CPR.
Remark that for the completeness direction, we have to ensure that a
sufficient number of constraints is provided for launching each rule
of the derivation.  To state the theorem we assume the following
notation ($n \cdot \BB$ stand for the scalar product of the multiset
$\BB$ by $n$):
\[
\P^n = \{  (r \at \KK 
\propag \GG  \mid  \CC,  n \cdot \BB) \mid   
(r \at \KK 
\propag \GG \mid  \CC, \BB) \in \P \}
\]%

\begin{theorem}[Soundness and completeness of coinductive semantics for CPR]
  \label{theorem:fixpoint:cpr}
  Let $\P$ be a CPR program and $n$ be some integer greater 
  than the maximal number of constraints occurring in the head of any 
  rule of $\P$. We have: 
\[
 \astate{n  \cdot  \EE}{\CC}{\bar X}  \in \smodel{\P^n}
  \text{ if and only if }
\prth{\EE \wedge \CC} \text{ is satisfiable within } \CP.
\]
\end{theorem}

  The proof of the theorem is based on the following completeness
  lemma, that ensures that to each intuitionistic deduction in the
  theory $\CP$ corresponds a CPR derivation \wrt\ the program $\P$.
  In the following, we use the notation $\mathcal T \vDash \EE \ra
  \exists X.\FF$ and $\EE \vdash_{\mathcal T} \exists X.\FF$ to
  emphasize that a deduction within a theory $\mathcal T$ is done in
  the classical logic model framework and the intuitionistic proof
  framework, respectively. Fortunately, both lines of reasoning
  coincide in our setting. This remarkable property is due to the
  fact we are reasoning in a fragment of classical logic, known as
  {\em coherent logic}, where classical provability coincides with
  intuitionistic provability.

\begin{lemma}[Intuitionistic completeness of CPR]
\label{proposition:cpr:completeness}
Let $\P$ be a CPR program and $n$ be some integer greater than the
maximal number of constraints occurring in the heads of $\P$. Let
$\EE$ and $\FF$ (resp\@. $\CC$ and $\DD$) be two multiset of CHR
(resp\@. built-in) constraints. If $\EE \wedge \CC  \vdash_{\CP} 
\exists \bar X.(\FF \wedge \DD)$, then there exist $\FF'$ and $\DD'$ such
that $\astate{n \cdot \EE}{\CC}{\fv(\FF, \CC)} \rapas[P^n] \astate{ n
  \cdot \FF'}{\DD'}{\bar Y}$ and $\FF'\wedge \DD' \vdash_{\CT} \exists
\bar X.(\FF \wedge \CC)$.
\end{lemma}

\begin{proof}[Proof sketch]
  By structural induction on the proof tree $\pi$ of $ \EE \wedge \CC
  \vdash_{\CP} \exists \bar X.(\FF \wedge \DD)$. For the case where $\pi$
  is a logical axiom, we use the reflexivity of $ \rapas[P^n]$. For
  the case where $\pi$ is a non-logical axiom from $\CT$, we use the
  definition of $\aequiv$.  For the case where $\pi$ is a non-logical
  axiom corresponding to a propagation rule $\KK \propag \GG \mid
  \BB_c , \BB_b$, we choose $\FF' = \prth{\EE, \BB_c}$ and $\DD' =
  \prth{\CC \wedge \GG \wedge \BB_b}$, and apply $r$. For the case
  where $\pi$ ends with a cut or a right introduction of a
  conjunction, we use induction hypothesis and the fact that
  constraints are never consumed along a CPR derivation.  Other cases
  are more straightforward.
\end{proof}

\begin{proof}[Proof sketch of \cref{theorem:fixpoint:cpr}]
As we have noted previously, $\smodel{\P^n}$ is the set of states that
cannot be rewritten to an inconsistent states. Hence it is sufficient
to prove:
\[
 \astate{n  \cdot  \EE}{\CC}{\bar X} \; \not \! \!\!   \rapas[\P^n] \0
  \text{ if and only if }
\CP\nvDash \neg \exists \prth{\EE \wedge \CC}
\]
or equivalently the contrapositive:
\[
\CP\vDash \forall \prth{\prth{\EE \wedge \CC} \rightarrow \false}
\text{ if and only if }
\astate{n \cdot \EE }{\CC}{\bar X}  \rapas[\P^n] \0
\]
The ``if'' direction is direct by soundness of CHR.
For the ``only if'' direction, since $\CP$ is a {\em coherent logic
  theory} (\ie\ a set of formulas of the form $\forall\prth{\FF
  \rightarrow \exists \bar X.\FF'}$, where both $\FF$ and $\FF'$ are
conjunctions of atomic propositions), it can be assumed without loss
of generality that $ \EE \wedge \CC \vdash_{\CP} \false$ (See
\citeANP{BC05lpar}'s work about coherent logic
\citeNN{BC05lpar}). The result is then direct, by
\cref{proposition:cpr:completeness}.
\end{proof}

  The coinductive semantics for CPR, has strong similarities with the
  fixpoint semantics of CLP \cite{JL87popl}.  Both are defined by
  fixpoint of somehow dual operators and fully abstract the logical
  meaning of programs. Nonetheless the coinductive semantics of CPR is
  not compositional. That is not a particular drawback of our
  semantics, since the logical semantics we characterize is neither
  compositional. Indeed, if the logical readings of two states are
  independently consistent, then one cannot ensure that so is their
  conjunction. %
  It should be noticed that this non-compositionality prevents the
  immediate cause operators to be defined over the $\CT$-base (\ie\
  the cross product of the set of CHR constraints and the set of
  conjunctions of built-in constraints) as it is done for CLP, and
  requires a definition over set of states.

\section{Transition system semantics for CHR with persistent constraints}
\label{sect:hybrid}

In this section, we aim at obtaining a fixpoint semantics for the
whole language. 
Nonetheless, one has to notice that the completeness
result for CSR needs, among other things, the termination of $\rapa$, while
the equivalent result for CPR is based on the monotonic evolution of
the constraints store along derivations. Hence combining naively CSR
and CPR will break both properties, leading consequently to an
incomplete model.
In order to provide an accurate fixpoint semantics to programs
combining both kinds of rules (meanwhile removing unsatisfactory
scalar product in the wording of CPR completeness), we introduce a
notion of {\em persistent constraints}, following ideas of
\citeN{BRF10iclp} for their semantics $\omega_!$.  Persistent
constraints are special CHR constraints acting as classical logic
statements (\ie\ they are not consumable and their multiplicity does
not matter). Since they act as linear logic statements (\ie\ they are
consumable and their multiplicity matters), usual CHR constraints are
called {\em linear}. Because it combines persistent and linear 
constraints in a slightly less transparent way that $\omega_!$, we 
call this semantics {\em  hybrid}, and note it $\omega_h$.
  Due to the space limitations, the (non-trivial) proofs of the
  section are omitted, but can be found in the extended version of this
  paper.

\subsection{Hybrid operational semantics $\omega_h$}
\label{subsection:hybrid:operational}

On the contrary of $\omega_!$, the kind of a constraint (linear or
persistent) in $\omega_h$, is not dynamically determined according the
type of rules from which it was produced, but statically fixed. Hence,
we will assume the set of CHR constraints symbols is divided in two:
the {\em linear} symbols and the {\em persistent} symbols. Naturally,
CHR constraints built from the linear (resp. persistent) symbols are
called linear (resp. persistent) constraints. A {\em hybrid rule} is 
a CHR rule where the kept head contains only persistent constraints and
the removed head contains only linear constraints.
We will denote by $\states_p$, the set of {\em purely persistent}
states (\ie\ states of the form $\state{\PP}{\CC}{\bar X}$ where $\PP$
is a set of persistent constraints).  $\P^s$ will refer to the set of
simplification rules of a hybrid program $\P$,
respectively.

The hybrid semantics is expressed as a particular instance of the
equivalence based semantics presented in
\cref{subsection:preliminaries:abstract}. It uses the abstract state
equivalence extended by a {\em contraction} rule enforcing the
impotency of persistent constraints.

\begin{definition}[Hybrid operational transition]
  The {\em hybrid equivalence} is the smallest relation, $\hequiv$,
  over states containing $\aequiv$ satisfying the following rule:
\[
 \state{c, c, \CC}{\DD}{\bar X} \hequiv \state{c, \CC}{\DD}{\bar X} \text{  if $c$ is a persistent constraints }
\]%
The {\em hybrid transition system} is defined as the tuple $(\Sigma, \raph)$.
\end{definition}

The hybrid programs are programs where propagation rules ``commute''
with simplification rules in the sense of abstract rewriting
systems~\cite{terese03}.  In other words, derivations can be permuted
so that simplification rules are fired first, and propagation rules
fire only when simplification rule firings are exhausted.  Indeed, the
syntactical restriction prevents the propagation head constraints to
be consumed by simplification rules, hence once a propagation rule is
applicable, then it will be so for ever. Of course, number of CHR
programs do not respect the hybrid syntax and therefore cannot be run
in our framework.  Nonetheless, what we loose with this restriction,
we compensate by pioneering a logically complete approach to solve the
problem of trivial non-termination (see next
\cref{theorem:Hybrid:Implementation:Completeness}).

\subsection{Hybrid transition system semantics}
\label{section:hybrid:transition}

We present the transition system semantics for hybrid programs. This
semantics is expressed by fixpoints over the hybrid transition
system. It is built using the immediate cause operators we have defined in
the previous section.

\begin{definition}[Hybrid transition system semantics]
  The {\em hybrid (transition system) semantics} of a hybrid program
  $\P$ is defined as:
\[
\hmodel{\P} = {\nu
      \X.\prth{\prea[\P]\prth{\X} \cap \mu \Y.\prth{\pree[\P^s]\prth{\Y} \cup
          \prth{\Sigma_p \setminus \eqclass{\0}} }}}
\]
\end{definition}

The theorem we give next states soundness and completeness of the hybrid
transition system semantics of confluent programs, provided the states
respect a data-sufficiency property.
In this paper, we do not address the problem of proving confluence of
hybrid programs, but claim it can be tackled by extending
straightforwardly the work of \citeN{AFM99constraints} or by
adequately instanstiating the notion of abstract critical pair we
proposed in a previous work \cite{HF07rta}.

\begin{definition}[Data-sufficient state]
  A hybrid state $\sigma$ is {\em data-sufficient} \wrt\ a hybrid
  program $\P$ if any state $\sigma'$ accessible form $\sigma$ can be
  simplified (\ie\ rewritten by $\P^s$) into a purely persistent sate%
    \ (\ie\ for any state $\sigma'\in\Sigma$, if $\sigma\raphs\sigma'$, then
    there exists a sate $\sigma''\in\Sigma_p$ \st\
    $\sigma'\raphs[\P^s]\sigma''$).
\end{definition}

This property ensures there is at least one computation where
propagation rules are applied only once all linear constraints have been
completely simplified. It is a natural extension of the eponymous
property for CSR \cite{AFM99constraints}. 

\removebrackets 
\begin{theorem}[\textup{(}Soundness and completeness of hybrid transition system
 semantics\textup{)}\hspace{-2cm}]
\label{theorem:Hybrid:FullAbstraction}
Let $\P$ be a hybrid program such that $\P^s$ is confluent.
Let $\state{\LL, \PP}{\EE}{\bar X}$ be a data-sufficient state \wrt\
$\P$. We have:
\[
\state{\LL, \PP}{\EE}{\bar X} \in \hmodel{\P} 
\text{ if and only if }
 \prth{\LL \wedge \PP \wedge \EE} \text{ is satisfiable within } \CP
\]
\end{theorem}

The following proposition states that it is sufficient to consider only one fair
derivation. 
This result is fundamental to allow the hybrid semantics to be
efficiently implemented.

\begin{definition}[Propagation fair derivation]
  A derivation $\sigma_0 \raph \sigma_1 \raph \dots $ is {\em
    propagation fair} if for any propagation redex
  $\redex{r}{\KK}{\EE}{\bar X}$ of a state $\sigma_i$ in the
  derivation,
  there exist two states $\sigma_j$, $\sigma_{j+1}$ such that  the transition
  from $\sigma_j$ to $\sigma_{j+1}$ is a propagation application where
  the reduced redex is identical or stronger to
  $\redex{r}{\KK}{\EE}{\bar X}$ (\ie\ the reduced redex is of the form
  $\redex{r}{\KK}{\FF}{\bar Y}$ with $\CT \vDash \cexists{\bar Y}\FF
  \rightarrow \cexists{\bar X}{\EE}$).
\end{definition}

\begin{proposition}[Soundness and completeness of propagation fair derivations]
\label{prop:Hybrid:Fairness}
Let $\P$ be a hybrid program such that $\P^s$ is confluent and
terminating. Let $\sigma$ be a data-sufficient state. $\sigma \in
\hmodel{\P}$ holds if and only if there is a consistent propagation
fair derivation starting from $\sigma$.
\end{proposition}

\subsection{Implementation of the hybrid semantics}
\label{subsect:hybrid:implementation}

We continue by addressing the question of implementing  the hybrid
semantics in a sound and complete way. For this purpose, we assume
without loss of generality that the constraint symbols $f/1$, $f/2$,
$a/2$, $c_f/1$, and $c_a/1$ are fresh \wrt\ the program $\P$ we
consider. The implementation of a hybrid program $\P$ consists in
a source-to-source translation $\P^\diamond$ intended to be executed
in the concrete semantics $\omega_p$. This transformation is given in
detail in \cref{fig:Set:Source2Source}.
In order to be executed an hybrid state $\sigma$ has to be translated
into a concrete state $\sigma^\diamond$ as follows: if $\LL$ and
$(c_1, \dots, c_n)$ are multisets of linear and persistent constraints
respectively, then $\hstate{\LL, d_1, \dots d_n}{\DD}{\VV}^\diamond =
\cstate{\LL, f(d_1), \dots, f(d_n), c_f(0),
  c_a(0)}{\emptyset}{\DD}{\emptyset}{\VV}{0}$

Before going further, let us give some intuition about the behaviour
of the translation. If a rule needs two occurrences of the same
persistent constraint, step 1 will insert an equivalent rule which
needs only one occurrence of the constraint. In the translation each
persistent constraint can be applied in three different successive
states: {\em fresh}, indicated by $f/1$, {\em frozen}, indicated by
$f/2$, and {\em alive}, indicated by $a/2$. Step 2 ensures, on the one
hand, that only alive constraints can be used to launch a rule, and on
the other hand, that the persistent constraints of the right-hand side
are inserted as fresh. Each frozen and alive constraint is associated
to a time stamp indicating the order in which it has been
asserted. The fresh constraints are time stamped and marked as frozen
as soon as possible by $\textit{stamp}$, the rule of highest priority
(the constraint $c_f/1$ indicating the next available time
stamp). Only if no other rule can be applied, the $\textit{unfreeze}$
rule turns the oldest frozen constraint into an alive constraint while
preserving its time stamp (the constraint $c_a/1$ indicating the next
constraint to be unfrozen). Rule $\textit{set}$ prevents trivial
loops, by removing the youngest occurrence of two identical persistent
constraints. From a proof point of view, the application of this last
rule corresponds to the detection of a cycle in a coinduction proof,
the persistent constraints representing coinduction hypothesises
\cite{BM96csli}.

\begin{figure}

\begin{minipage}{0.90\textwidth}
\noindent
Let $\cal P^\diamond$ the program $\P$ where simplification  and
propagation rules are given with the priorities $3$ and $4$
respectively. Apply the following steps:
 \begin{description}
\item[step 1].
 Apply the following rules until convergence :\\
  If $(p:: c, d, \KK \backslash \HH \simplif \GG \mid \BB,
  \LL, \PP))$ is in $\cal P^\diamond$, 
  with $\CT \vDash \exists (c\!=\!d)$  \\then add the rule $(p:: c , \KK \backslash \HH
  \simplif c\! = \!d \wedge \GG \mid \BB, \LL, \PP)$ to $\cal P^\diamond$
  \smallskip
\item[step 2] Substitute any rule $(p:: c_1, \dots, c_m \backslash \HH \simplif \GG \mid \BB, \LL, d_1, \dots d_n)$   by \\
$(p::  a(X_1, c_1), \dots, a(X_m, c_m) \backslash \HH \simplif
 \GG \mid \BB_l,  \LL, f(d_1), \dots, f(d_n))$ \\
 where $x_1, \dots  x_m$ are pairwise distinct variables
\smallskip
\item[step 3] Add to $\cal P^\diamond$ the rules:\\
$
\begin{array}{llrcl}
 1 &:: &\textit{stamp} &@ & f(X), c_f(Y) \simplif f(Y, X), c_f(Y+1) \\
 2 &:: &\textit{set} &@ & a(Y, X) ~ \backslash ~ a(Z, X) \simplif Y <  Z \mid \top \\
 5 &:: &\textit{unfreeze} &@ & f(Y, X), c_a(Y)  \simplif  a(Y, X), c_a(Y+1) 
\end{array}
$
 \end{description}
\end{minipage}

\caption{Source-to-source translation for hybrid programs}
\label{fig:Set:Source2Source}
\end{figure}

The two following theorems state that our implementation is sound and
complete \wrt\ failure.
\Cref{theorem:Hybrid:Implementation:Soundness} shows furthermore that the
implementation we propose here is sound \wrt\ finite
success. It is worth noting that it is hopeless to look for a
complete implementation \wrt\ to success, since the problem to
know if a data-sufficient state is in the coinductive semantics is
undecidable. The intuition behind this  claim is that otherwise it
would possible to solve the halting problem.

\begin{theorem}[Soundness \wrt\ success and failure]%
\label{theorem:Hybrid:Implementation:Soundness}
Let $\P$ be a hybrid program such that $\P^s$ is confluent,
and let
$\hstate{\CC}{\EE}{\bar X}^\diamond $ $ \ra_{\P^\diamond}^* $
$\cstate{\emptyset}{\DD}{\FF}{T}{\bar Y}{i} \not \ra_{\P^\diamond}$ be a
terminating derivation. 
$\hstate{\CC}{\EE}{\bar X} \in \hmodel{\P}$ holds
if and only if 
$\FF$ is satisfiable within $\CT$.
\end{theorem}

\begin{theorem}[Completeness \wrt\ failure]
\label{theorem:Hybrid:Implementation:Completeness}
Let $\P$ be a hybrid program such that $\P^s$ is confluent and
terminating. Let $\hstate{\CC}{\EE}{\bar X}$ be a data-sufficient state.
If $\hstate{\CC}{\DD}{\bar X} \notin \hmodel{\P}$, then any concrete
derivation starting form $\hstate{\CC}{\DD}{\bar X}^\diamond$ finitely
fails.
\end{theorem}

The implementation we propose has strong connections with the
co-SLD, an implementation of the \gfp\ semantics of Logic Programming
proposed by \citeN{SMBG06iclp}.  Both are based on a dynamic
synthesis of coinductive hypothesises and a cycle detection in
proofs. But because it is limited to rational recursion, the co-SLD is
logically incomplete \wrt\ both successes and failures (\ie\ there are
queries true and false \wrt\ the logical reading of a
program that cause the interpreter to loop).
 That contrasts with CHR, where any coherent constraint system can be
 used without loosing logical completeness \wrt\ failures.

\section{Applications}
\label{sect:applications}

In this section, we illustrate the power of CHR for coinductive
reasoning when it is provided with its fixpoint semantics.  In
particular we show it yields an elegant framework to realize coinductive
equality proofs for languages and regular expressions as presented by
\citeN{Rutten98concur}.

\subsection{Coinductive language equality proof}

Firstly, let us introduce the classical notion of binary automaton in
a slightly different way from usual.  A {\em binary automaton} is a
pair $(\L, f)$ where $\L$ is a possibly infinite set of {\em states}
and $f : \L \ra \{0,1\} \times \L \times \L$ is a function called
{\em destructor}. Let us assume some automaton $(\L, f)$. For any state
$L\in \L$ such that $f(L)=(T, L_a, L_b)$, we  write $L
\xrightarrow{{}_a} L_a$, and $L \xrightarrow{{}_b} L_b$, and $t(L) =
T$.  $\L(L) = \{a_1 \dots a_n | L \xrightarrow{{}_{a_1}} L_1
\xrightarrow{{}_{a_2}} \dots \xrightarrow{{}_{a_n}} L_n
\wedge t (L_n) = 1 \}$ is the language accepted by a state $L$.  A {\em
  bisimulation} between states is a relation $\R \subset \L \times \L$
verifying:
\[
\text{If } K\,\R\, L \text{ then } 
\begin{cases}
t(K) = t(L), \\
K  \xrightarrow{{}_a} K_a,\, L  \xrightarrow{{}_a} L_a,\,   K_a\,\R\, L_a, \text{ and } \\
K  \xrightarrow{{}_b} K_b,\, L  \xrightarrow{{}_b} L_b,\,   K_b\,\R\, L_b,
\end{cases} 
\]

Contrary to the standard definition, in the present setting,
an automaton does not have an initial state and may have an infinite number
of states. As represented here, an automaton is a particular coalgebra
\cite{BM96csli}.
 Due to the space limitations, we will not enter in
details in the topic of coalgebra\footnote{We invite unfamiliar
  readers to refer to the gentle introduction of \citeN{Rutten98concur}.}, but only state the following {\em
  Coinductive Proof Principle} \cite{Rutten98concur,BM96csli} which
gives rise to the representation of automata as coalgebra:

\begin{center}
  \begin{minipage}{0.9\linewidth}
    In order to prove the equality of the languages recognized by two
    states $K$ and $L$, it is sufficient to establish the existence of
    a bisimulation relation in $\L$ that includes the pair $(K, L)$.
  \end{minipage}
\end{center}

A nice application of CPR consists in the possibility to directly
represent coalgebra and prove bisimulation between states.  For
instance, one can easily represent a finite automaton using variables
for states and binary user-defined constraints (of main symbol $f/2$)
for the destructor function. \Cref{fig:Automaton} gives an example of
an automaton and its representation as a multiset $\DD$ of CHR
constraints.
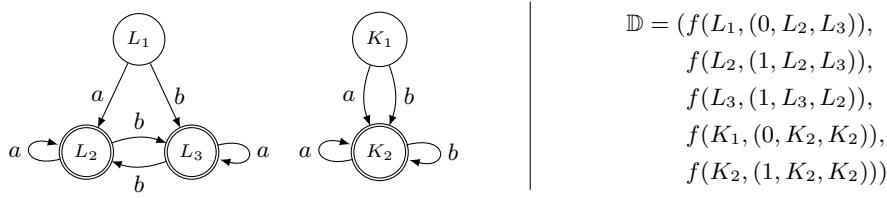
\begin{figure}
  \begin{tikzpicture}[exen/.style={shape=circle,draw}, >=latex, shape border  uses incircle, final/.style={exen,double}]
\node[exen](L1) at (0.8,1.5) {\scriptsize $L_1$};
\node[final](L2) at (0.1,0) {\scriptsize $L_2$} edge [loop left] node {$a$} ();
\node[final](L3) at (1.5,0) {\scriptsize $L_3$} edge [loop right] node {$a$} ();
\node[exen](K1) at (4,1.5) {\scriptsize $K_1$};
\node[final](K2) at (4,0) {\scriptsize $K_2$} edge [loop right] node {$b$} () edge [loop left] node {$a$} ();

\draw[->](L1) to node [left] {$a$} (L2);
\draw[->](L1) to  node [right] {$b$} (L3);
\draw[->](L2) to [out=20,in=160] node [above] {$b$} (L3);
\draw[->](L3) to [out=200,in=340] node [below] {$b$} (L2);

\draw[->](K1) to [out=250,in=110] node [left] {$a$} (K2);
\draw[->](K1) to [out=290,in=70] node [right] {$b$} (K2);

\draw (6,2) -- +(0,-2.5);
\node [below] at (9, 2) {
$
  \begin{aligned}
\DD =  (& 
f(L_1, (0, L_2, L_3)), \, \\& f(L_2, (1, L_2, L_3)),\,  \\&f(L_3, (1, L_3, L_2)),  \\ 
& f(K_1, (0, K_2, K_2)), \,  \\& f(K_2, (1, K_2, K_2)) )
  \end{aligned}
$};
  \end{tikzpicture}
\caption{A binary automaton and its CPR representation.} 
\label{fig:Automaton}
\end{figure}
Once the automaton representation is fixed, one can translate the
definition of bisimulation into a single propagation rule:
\[
f(L, (L_t, L_a, L_b)), \, f(K, (K_t, K_a, K_b)),\, L\! \sim\! K \propag
L_t \! =\! K_t, \, L_a \!\sim\! K_a, \, L_b\! \sim\! K_b
\]
Using coinductive proof principle and \cref{theorem:fixpoint:cpr}, it
is simple to prove two states of the coalgebra represented by
$\DD$ accept or not the same language. For example, to conclude that
$L_1$ and $K_1$ recognize the same language while $L_1$ and $K_2$ do
not, one can prove the execution of $\hstate{3\cdot\prth{\DD, L_1 \sim
    K_1}}{\top}{\emptyset}$ never reaches inconsistent states, while
there are inconsistent derivations starting from
$\hstate{3\cdot\prth{\DD, L_1 \sim K_2}}{\top}{\emptyset}$.

\subsection{Coinductive solver for regular expressions}
\label{subsect:application:solver}

We have just shown that CPR yields a nice framework for coinductive
reasoning about coalgebra. Nonetheless the explicit representation of
an automaton by user defined constraints (as in
\cref{fig:Automaton}) would limit ourselves to finite state
automata. One simple idea to circumvent this limitation is to implictly represent
infinite states automata. For instance, one can represent
states using regular expressions and implement the computation of the
destructor function using derivatives \cite{Rutten98concur}.

Let us assume the following syntax for regular expressions:
\[
E ::= L \mid a \mid b  \mid  E,E \mid E^*
 \qquad L ::= \ezero \mid \esum{E}{L}
\]
where $a$ and $b$ are characters, $({}^*)$ and $(,)$ stand for
the Kleene star and the concatenation operators respectively, and a list
corresponds to the alternation 
of its elements. Here follows a possible implementation of the
destructor function%
  \footnote{A complete version of program can be found in technical
    report version of the paper.}%
:
\[
\begin{array}{l}
f(\ezero, R)  \;\simplif\; R = (0, \ezero, \ezero).\\
f(a, R) \;\simplif\; R = (0, \esum{\eone}{}, \ezero).\\
f(b, R) \;\simplif\; R = (0, \ezero, \esum{\eone}{}).\\
f(\esum{E}{L}, R) \;\simplif\; R = (T, A, B), f(E, (E_t, E_a, E_b)), f(L, (L_t, L_a, L_b)), \\
\null \qquad \qquad  or(E_t, L_t, T), merge(E_a, L_a, A), merge(E_b, L_b, B).\\
f(E^*, R) \;\simplif\; R = (1, \esum{\prth{E_a,\esum{E^*}{}}}{}, \esum{(E_b,\esum{E^*}{})}{}), f(E, (\_, E_a, E_b)). \\
f((E, F), R) \;\simplif\;  f(E, (E_t, E_a, E_b)), f_\text{conc}(E_t, E_a, E_b, F, R).\\
f_\text{conc}(0, E_a, E_b, F, R) \;\simplif\; R = (0, \esum{(E_a,F)}{}, \esum{(E_b,F)}{}).\\
f_\text{conc}(1, E_a, E_b, F, R) \;\simplif\; R = (T, A, B), f(F, F_a, F_b)),\\ 
\null \qquad \qquad  merge(\esum{(E_a, F)}{}, F_a, A), merge(\esum{(E_b, F)}{}, F_b, B).
\end{array}
\]
where $or/3$ unifies its third argument with the Boolean disjunction
of its two first elements and $merge/3$ unifies its last argument
with the ordered union of the lists given as first arguments. Now one
can adapt the encoding of bisimulation given in the previous subsection
as follows:{%
\[
L\! \sim\!\!\hspace{0.5px} K \!\propag\! nonvar(L),\!\hspace{0.5px} nonvar(K) | f\!\hspace{0.5px} (L,\!\hspace{0.5px} (T,\!\hspace{0.5px}{} L_a,\!\hspace{0.5px}{} L_b)\!),\!\hspace{0.5px} f\!\hspace{0.5px} (K, (T,\!\hspace{0.5px}{} K_a,\!\hspace{0.5px}{} K_b)\!),\!\hspace{0.5px}{}
 L_a \!\sim\!\!\hspace{0.5px} K_a,\!\hspace{0.5px}{}  L_b\! \sim\!\!\hspace{0.5px} K_b. 
 \]%
}
We are now able to prove equality of regular expression using the
implementation of CHR hybrid semantics provided in
\cref{subsect:hybrid:implementation}. For example, the following state
leads to an irreducible consistent state. This implies thanks to the
Coinductive Proof Principle together with 
\cref{theorem:Hybrid:FullAbstraction} and
\cref{theorem:Hybrid:Implementation:Soundness} that the two regular
expressions recognize the same language/
\[
  \hstate{((b^*,a)^*,(a,b^*))^* \sim [[]^*, (a, [a,b]^*), ([a,b]^*,(a,(a,[a,b]^*)))]}{\top}{\emptyset}^\diamond
\]

It should be underlined that the use of simplification rules instead
of propagation rules for encoding the destructor function is essential
here in order to avoid rapid saturation of the memory by useless
constraints. Notice that on the one hand, the confluence of the set of
simplification rules needed by the theorems of \cref{sect:hybrid} can
be easily inferred, since the program is deterministic.  On the other
hand, termination of the the set of simplification rules, which is
required by \cref{theorem:Hybrid:Implementation:Completeness} can be
easily established by using, for instance, techniques we have
recently proposed for single-headed programs~\cite{HLH11ppdp}.

Of course, no one should be surprised that equivalence of regular
expressions is decidable. The interesting point here is that the
notion of coalgebra and bisimulation can be casted naturally in
CHR. Moreover, it is worth noticing the program given has the
properties required by a constraint solver.  Firstly the program is
effective, \ie\ it can actually prove or disprove if that two
expressions are equal. The first part of the claim can be proved using
Kleene theorem \cite{Rutten98concur} and the idempotency and
commutativity of the alternation, enforced here by the
$\textit{merge/3}$ predicate. The second part is direct by the
completeness \wrt\ failures. Secondly, the program is incremental: it
can deal with partially instanciated expressions by freezing some
computations provided without enough information. Last, but not 
least, one can easily add to the system new expressions (as for
instance $\epsilon$, $E?$, or $E^+$). For this purpose it is just
necessary to provide a new simplification rule for computing the
result of the  corresponding destructor function. For example, we can
add to the program the following rule and prove as previously that
$a^+$ and $(a, a^*)$ recognize the same language while $a^+$ and $a^*$
do not:
\[
f(K^+, R) \simplif R = (T,\, [K_a, (K_a,K^+)],\, [K_b, (K_b, K^+)]),\, f(K,\, (T,\, K_a,\, K_b)).
\]

\section{Conclusion}
\label{sect:conclusion}

We have defined a \lfp\ semantics for CHR simplification rules and a
\gfp\ semantics for CHR propagations rules, and proved both to be
accurate \wrt\ the logical reading of the rules. By using a
hybrid operational semantics with persistent constraints similar to
the one of Betz et al\@., we were able to characterize CHR programs
combining both simplification and propagation rules by a fixpoint semantics
without losing completeness \wrt\ to logical semantics. In
doing so, we have improved noticeably results about logical semantics
of CHR. 
Subsequently we proposed an implementation of this hybrid semantics
and showed it yields an elegant framework for programming with coinductive
reasoning.

The observation that non-termination of all derivations starting from
a given state ensures this latter to be in the coinductive semantics
of an hybrid program, suggests that the statics analysis of universal
non-termination of a CHR program might be worth investigating.  The
comparison of CHR to other coinductive programming frameworks such
that the circular coinductive rewriting of \citeN{GLR00ase} may
suggest it should be possible to improve completeness with respect to
success of the implementation proposed here.

\subsection*{Acknowledgements} 

{

  The research leading to these results has received funding from the
  Madrid Regional Government under the CM project P2009/TIC/1465
  (PROMETIDOS), the Spanish Ministry of Science under the MEC project
  TIN-2008-05624 {\em DOVES}, and the European Seventh Framework
  Programme FP/2007-2013 under grant agreement 215483 (S-CUBE).

We are grateful to C\'esar S\'anchez, for interesting discussions
about coinduction, and to Santiago Zanella B\'eguelin, for referring
us to \citeANP{BC05lpar}'s work.
At the end, we would like to thank reviewers for their helpful and
constructive comments.

}

\bibliographystyle{acmtrans}
\bibliography{biblio}

\end{document}